\documentclass[letterpaper, 10 pt, conference]{ieeeconf}

\usepackage{graphicx}
\usepackage{amsmath}

\usepackage{amsthm}
\usepackage{amssymb} 
\usepackage{comment}
\usepackage{color}
\theoremstyle{definition}

\newtheorem{defn}{Definition}
\newtheorem{thm}{Theorem}

\newtheorem{assum}{Assumption}
\newtheorem{rem}{Remark}
\newtheorem{prob}{Problem}
\newtheorem{lem}{Lemma}
\newtheorem{cor}{Corollary}
\newtheorem{exmp}{Example}

\DeclareMathOperator*{\argmax}{arg\,max}

\title{\LARGE \bf The Stackelberg Equilibrium for One-sided Zero-sum Partially Observable Stochastic Games} 

\author{Wei Zheng, Taeho Jung, Hai Lin
\thanks{The support of the National Science Foundation (Grant No. IIS-1724070, CNS-1830335, IIS-2007949) is gratefully acknowledged. Corresponding author Hai Lin. Tel. +574-6313177.}
\thanks{Wei Zheng and Hai Lin are with the Department of Electrical Engineering, University of Notre Dame, Notre Dame, IN, 46556 USA. {\tt\small wzheng1@nd.edu, hlin1@nd.edu}. Taeho Jung is with the Computer Science and Engineering Department, University of Notre Dame, IN, USA. {\tt\small tjung@nd.edu}.}
}

\begin{document}

\maketitle
\thispagestyle{empty}
\pagestyle{empty}

\begin{abstract}
Formulating cyber-security problems with attackers and defenders as a partially observable stochastic game has become a trend recently. Among them, the one-sided two-player zero-sum partially observable stochastic game (OTZ-POSG) has emerged as a popular model because it allows players to compete for multiple stages based on partial knowledge of the system. All existing work on OTZ-POSG has focused on the simultaneous move scenario and assumed that one player's actions are private in the execution process. However, this assumption may become questionable since one player's action may be detected by the opponent through deploying action detection strategies. Hence, in this paper, we propose a turn-based OTZ-POSG with the assumption of public actions and investigate the existence and properties of a Stackelberg equilibrium for this game. We first prove the existence of the Stackelberg equilibrium for the one-stage case and show that the one-stage game can be converted into a linear-fractional programming problem and therefore solved by linear programming. For multiple stages, the main challenge is the information leakage issue as the public run-time action reveals certain private information to the opponent and allows the opponent to achieve more rewards in the future. To deal with this issue, we adopt the concept of $\epsilon$-Stackelberg equilibrium and prove that this equilibrium can be achieved for finite-horizon OTZ-POSGs. We propose a space partition approach to solve the game iteratively and show that the value function of the leader is piece-wise linear and the value function of the follower is piece-wise constant for multiple stages. Finally, examples are given to illustrate the space partition approach and show that value functions are piece-wise linear and piece-wise constant. 
\end{abstract}


\section{Introduction}\label{sec:introduction}
With advances of technologies in computing, communications, and control, new engineered systems require tighter and tighter integration of cyber systems and physical systems, which increases security risks and attack surfaces, and therefore brings new challenges to cyber-security defense \cite{humayed2017cyber}. As attack surfaces increase, cyber-attacks may be composed of multiple stages. For example, attackers may compromise the most vulnerable device first and then exploit it for attacking other devices. Through observing attack effects and behaviors of defenders, attackers may adjust their attacking strategies dynamically. Meanwhile, due to the limited resources for anomaly detection, information such as infected devices is usually private for the defender. The defender has to infer this information from other observations of the system. The dynamic nature and the partial observability of this kind of attack require the defender to be more reactive to the system and robust to information uncertainty.

Recently, various approaches have been proposed to mitigate cyber-security concerns such as the machine learning approach \cite{salloum2020machine}, the data mining approach \cite{buczak2015survey} and the game theory approach \cite{do2017game}. In this paper, we will focus on the game theory approach because it provides a theoretical study of interactions among independent players. Most of the existing games in the literature focus on static behaviors and ignore the dynamic nature. Dynamic games such as repeated games \cite{fallah2008puzzle}, evolutionary games \cite{huang2017markov} and stochastic games \cite{lalropuia2019modeling} consider the dynamic behavior but assume the full observability of game information. As a critical branch of the game theory, the partially observable stochastic game (POSG) attracts more and more attention in the area because the POSG allows players to compete sequentially based on partial knowledge of the system, which is closer to real cyber-security problems \cite{luo2009game,horak2019optimizing,wang2019players}. Game models such as the static game, repeated game, and stochastic game are special cases of the POSG. However, solving the general POSG is a nontrivial task. The dynamic programming algorithm that solves the game exactly becomes inefficient quickly beyond a small horizon \cite{hansen2004dynamic}. Although approximate algorithms such as the $\epsilon$-pruning approximation \cite{kumar2009dynamic} and the Bayesian game approximation \cite{emery2004approximate} have been proposed, the POSG planning algorithm is still not mature enough for practical applications. 

Sub-classes of POSGs have been proved more practical than the general case. The one-sided two-player zero-sum POSG assumes that one player can observe the state directly while the other player accesses the state via a partial observation \cite{horak2016point}. This model allows both players to maintain a common belief over states, which enables us to design and implement efficient planning algorithms. Meanwhile, this model reserves properties of the general POSG, such as the dynamic nature and partial observability. A point-based approximate algorithm \cite{horak2016point}, a heuristic search algorithm \cite{horak2017heuristic}, and a mixed-integer linear programming approach \cite{ahmadi2018partially} have been proposed to solve the game efficiently. Beyond the theoretical work, the OTZ-POSG has been extensively discussed for cyber-security problems such as the computer network defense \cite{horak2019optimizing,tsemogne2020partially,tomavsek2020using}. For a more general setup than the OTZ-POSG \cite{horak2016point}, the two-player zero-sum POSG with public observations assumes that each player has private information and has a partial observation on the other player's private information \cite{horak2019solving}. By assuming that observations are public, the existence of the Nash equilibrium is guaranteed.

All aforementioned OTZ-POSGs assume that both players move simultaneously and one player's action is unobservable by the other player. However, these assumptions may become questionable for real-world applications. First, to guarantee the simultaneous move is not always realistic, especially for a competing scenario. Secondly, the actions of players may become observable as opponents may deploy action detection strategies in the execution stage. Hence, we propose the turn-based OTZ-POSG with public actions and investigate the finite-horizon Stackelberg equilibrium in this paper. In this game, one player (the leader/the defender) plays first while the other player (the follower/the attacker) follows. Certain statues of the environment or the attacker, such as the attacker's locations, are only partially observed by the defender. The defender plays first because he has to deploy defending resources according to partial observations before the attack. In each stage, both actions of the defender and the attacker are observable at the end of the stage.

First, we prove the existence of the Stackelberg equilibrium for the one-stage OTZ-POSG and show that the one-stage OTZ-POSG can be converted into a linear-fractional programming problem, and therefore solved by linear programming. The value function of the leader is piece-wise linear and the value function of the follower is piece-wise constant. These value functions are solved by enumerating all extreme points of the linear program. For multiple stages, the main challenge is the information leakage issue because the follower's policy is private-information-dependent. When taking full advantage of the private information, the follower reveals certain private information to the leader. Then, the leader can infer more private information from the follower and achieves more rewards in the following stages. To solve this issue, we adopt the concept of $\epsilon$-Stackelberg equilibrium \cite{bacsar1998dynamic}. At this equilibrium, the follower sacrifices certain rewards in the current stage for more rewards in the following stages. We propose a space partition approach to solve the game through value iteration and show that value functions for both players are piece-wise linear and piece-wise constant respectively.

The main contribution of this paper is twofold. First, we prove the existence of the Stackelberg equilibrium for the one-stage OTZ-POSG and show that the one-stage game can be solved by linear programming. Hence, value functions of players are piece-wise linear and piece-wise constant respectively. Secondly, we adopt the concept of $\epsilon$-Stackelberg equilibrium and prove that the $\epsilon$-Stackelberg equilibrium is achieved for finite-horizon OTZ-POSGs with public actions. Meanwhile, we propose a dynamic programming algorithm to solve the finite-horizon OTZ-POSG iteratively through belief space partition.

The rest of the paper is organized as follows. Section \ref{sec:pro_formula} defines the OTZ-POSG and formulates the problem. The existence of the Stackelberg equilibrium for the one-stage OTZ-POSG and the policy-solving algorithm are given in Section \ref{sec:one_stage_game}. Section \ref{sec:multiple_stage} introduces the $\epsilon$-Stackelberg equilibrium and proposes the space partition approach to solve the multiple-stage game iteratively. Section \ref{sec:conc} concludes the paper.

\textbf{Notations}: $\mathbb{R}$ represents the set of real numbers and $\mathbb{R}^{p \times q}$ represents the set of real-valued matrices with $p$ rows and $q$ columns. Specifically, $\mathbf{1}_p$ represents a vector of ones with dimension $p$ and $\mathbf{0}_p$ represents a vector of zeros with dimension $p$. For a vector $v \in \mathbb{R}^{p \times 1}$, $v^T$ stands for the transpose of the vector and $v_{i}$ stands for the $i^{th}$ element of the vector $v$. For two vectors $v, u \in \mathbb{R}^{p \times 1}$, $v \leq u$ implies $v_i \leq u_i$ for all $i$. For a finite set $S$, $|S|$ represents the cardinality of the set. $\mathcal{P}(\cdot)$ denotes the probability and $\mathbb{E}$ denotes the expectation.

\section{Preliminaries and Problem Formulations}\label{sec:pro_formula}
In this section, we give the formal definition of the OTZ-POSG and formulate the problem. 
\begin{defn}
A OTZ-POSG model is defined as a tuple $\mathcal{G} = (I, S, O, A^L, A^F, T, \Xi, \Upsilon, b_0)$ where
\begin{itemize}
    \item $I = \{leader(L),follower(F)\}$ is a set of players;
    \item $S$ is a finite set of states;
    \item $O$ is a finite set of observations; 
    \item $A^i$ is a finite set of actions of player $i \in I$;  
    \item $T: S \times A^L \times A^F \times S \to [0,1]$ is a transition function;
    \item $\Xi: S \times O \to [0,1]$ is an observation function;
    \item $\Upsilon: S \times A^L \times A^F \to \mathbb{R}$ is a reward function of player $F$;
    \item $b_0: S \to [0,1]$ is an initial belief over states.  
\end{itemize}
\end{defn}

The game $\mathcal{G}$ is played in turn and actions of players are public. The game playing process is shown in Fig.~\ref{fig:information_flow}. Arrows represent information dependencies. For any stage $t$, the state $s_t \in S$ is only informed to the follower. The leader takes an action first according to the observation $o_{t} \in O$. This action would not be revealed to the follower until the follower's action is taken. Each player achieves a reward and this reward is not explicitly announced until the end of the game. The state of the system transits from $s_t$ to $s_{t+1}$ according to the transition function $T(s_{t+1}|s_t,a^L_t,a^F_t)$ which defines the distribution over the next state $s_{t+1} \in S$ after taking a joint action $[a^L_t,a^F_t] \in A^L \times A^F$ from the state $s_t$. An observation $o_{t+1} \in O$ generated according to the observation function $\Xi(o_{t+1}|s_{t+1})$ is publicly observed by players. Since actions are public, the leader's optimal strategy in each stage relies on the action executed by the follower in the previous stage (see the red dotted line). The initial belief $b_0$ is a probabilistic distribution over states. It is used to describe the initial knowledge of players. Initially, a state $s_0 \in S$ is drawn according to this distribution and the state is only informed to the follower. We assume that the initial belief is common knowledge of both players.

\begin{figure}
\centering
\includegraphics[width=0.9\linewidth]{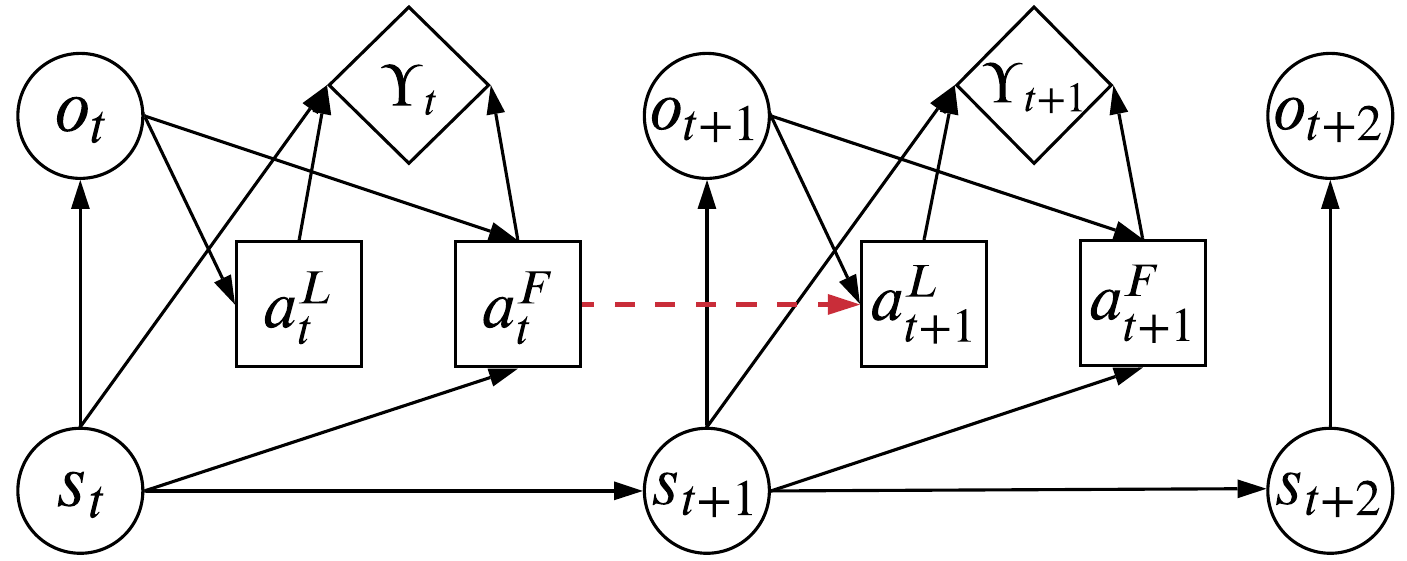}\vspace{-5pt}
\caption{The information flow of the OTZ-POSG.}
\label{fig:information_flow}
\end{figure}

Since the game is zero-sum, the reward function of the leader is $-\Upsilon(s,a^L,a^F)$ for all $s \in S$, $a^L \in A^L$ and $a^F \in A^F$. Because the game is one-sided, the leader has to infer the state $s_t$ of the game through the initial belief $b_0$ and the information observed.

\begin{defn}
Up to stage $t$, the observable path of the leader is $\upsilon_t = b_0 a^L_0 a^F_0 o_1...a^L_{t-1} a^F_{t-1} o_t$ and the observable path of the follower is $\omega_t = b_0 s_0 a^L_0 a^F_0 s_1 o_1...a^L_{t-1} a^F_{t-1} s_t o_t$ where $a_\tau^L \in A^L$, $a_\tau^F \in A^F$, $s_{\tau} \in S$, and $o_{\tau} \in O$ for all $0 \leq \tau \leq t$. 
\end{defn}

Based on the observable path, the leader can reason about the state of the game through a probability distribution over states for which we call the belief state.

\begin{defn} A belief state of a OTZ-POSG $\mathcal{G}$ is defined as a conditional probability distribution over states, i.e., $b_{t}^{s_t} = \mathcal{P}(s_t | b_0,a^L_0,a^F_0,...,o_{t-1},a^L_{t-1},a^F_{t-1},o_{t})$. 
\end{defn}

Beginning with the initial belief $b_0$, we can calculate the belief $b_t$ by the Bayes' rules incrementally. The updated belief $b_{t+1}$ after taking a particular joint action $a_t = [a^L_t, a^F_t]$ and observing $o_{t+1}$ is 
\begin{equation}\label{eq:belief_updating}
    b_{t+1}^{s_{t+1}} = \frac{\Xi^{o_{t+1}}_{s_{t+1}} \sum_{s_{t}} T^{s_{t+1}}_{s_{t},a_t} b_{t}^{s_{t}}}
    {\sum_{s_{t+1}} \Xi^{o_{t+1}}_{s_{t+1}} \sum_{s_{t}} T^{s_{t+1}}_{s_{t},a_t} b_{t}^{s_{t}}},
\end{equation}
where $T^{s_{t+1}}_{s_{t}, a_t}$ and $\Xi^{o_{t+1}}_{s_{t+1}}$ are concise notations for transition probabilities and observation probabilities.

To behave optimally, players have to plan to act according to their observable paths. Because the belief is a sufficient statistic of the path $\upsilon_t$, the leader can act equivalently according to the belief \cite{aberdeen2007policy}. As $\upsilon_t$ is a sub-sequence of $\omega_t$, the follower can maintain the same belief state as that of the leader, and therefore, act equivalently according to the belief-state pair. In this paper, we consider mixed policies for players. A mixed policy is a probability distribution over the action space.

\begin{defn}
The policy of the leader is defined as a mapping from a belief to a distribution $\eta_t$ over the action space $A^L$, i.e., $\pi^L : b_t \to \eta_t$ and the policy of the follower is defined as a mapping from a belief-state pair to a distribution $\delta_t$ over action space $A^F$, i.e., $\pi^F : (b_t, s_t) \to \delta_t$ where $\eta_t \in \mathbb{R}^{|A^L| \times 1}$ and $\delta_t \in \mathbb{R}^{|A^F| \times 1}$. 
\end{defn}

\begin{assum}\label{assum:leader_observed}
In each stage, the strategy adopted by the leader is known by the follower. 
\end{assum}

\begin{rem}
This assumption is usually referred to the \textit{commitment} in the literature \cite{conitzer2006computing}. We have this assumption because the leader usually arrives at the site where the game is played before the follower. For example, the defender usually arrives at the site before the attacker for cyber-security defense. The strategy adopted by the leader could be learned by the follower through long-term observations. 
\end{rem}

Once the policies of players are fixed, each player is expected to receive a reward for finite stages.

\begin{defn} Given a OTZ-POSG model $\mathcal{G}$, a finite horizon $h$, and a pair of policies $[\pi^L, \pi^F]$ for players, the total reward achieved by the leader is defined as 
\begin{equation}
    v_{\pi^{L}, \pi^{F}}^{L} (b_0) =\mathbb{E}[\sum_{t = 1}^{h} \Upsilon(s_t,a^L_t,a^F_t)|b_0, \pi^{L}, \pi^{F}]
\end{equation}
and the total reward achieved by the follower is defined as 
\begin{equation}
    v_{\pi^{L}, \pi^{F}}^{F}(b_0, s_0) = \mathbb{E} [\sum_{t = 1}^{h} \Upsilon(s_t,a^L_t,a^F_t)|b_0, s_0, \pi^{L}, \pi^{F}]
\end{equation}
\end{defn}

\begin{rem}
Once the policies of players are fixed, both the state transition and state observation are stochastic. Hence, total rewards are defined as expectations of cumulative rewards over all stages. It is noted that the total rewards defined above are for the reward function of the follower. Maximizing the reward for the leader is equivalent to minimizing the reward $v_{\pi^{L}, \pi^{F}}^{L} (b_0)$. 
\end{rem}

Since the game is turn-based and the policy of the leader is known by the follower, the leader has to optimize the total reward concerning the best response of the follower. Hence, we introduce the Stackelberg equilibrium to study the behavior of the game \cite{breton1988sequential}.

\begin{defn} Given the total rewards for both players, a pair of policies $[\hat{\pi}^{L}, \hat{\pi}^{F}]$ forms a Stackelberg equilibrium if they satisfy following conditions, 
\begin{equation}
\begin{array}{lll}
    v^L_{\hat{\pi}^{L}, \sigma(\hat{\pi}^{L})}(b_0)  & \leq v^L_{\pi^{L}, \sigma({\pi}^{L})}(b_0), &\forall \pi^{L}, \\
    v^F_{\hat{\pi}^{L}, \hat{\pi}^{F}}(b_0,s_0) & \geq v^F_{\hat{\pi}^{L}, \pi^{F}}(b_0, s_0), &\forall \pi^{F}, 
\end{array}
\end{equation}
where $\sigma({\pi}^{L})$ is a reaction function of the follower defined by $\sigma({\pi}^{L}) = \argmax_{\pi^F} v^F_{{\pi}^{L}, {\pi}^{F}}(b_0,s_0)$.
\end{defn}

At the Stackelberg equilibrium, neither the leader nor the follower has the incentive to change the policy. Because of Assumption \ref{assum:leader_observed}, the follower always responses optimally to the leader's policy. Hence, the Stackelberg equilibrium defines its first inequality with the reaction function $\sigma({\pi}^{L})$.

In this paper, we will study the existence of the Stackelberg equilibrium and provide policy solving algorithms for the one-stage OTZ-POSG and its $\epsilon$-version for the multi-stage OTZ-POSG.

\begin{prob}\label{problem:main_problem} 
Given a OTZ-POSG model $\mathcal{G}$ and a finite horizon $h$, solve policies $\pi^L$ and $\pi^F$ that achieve the Stackelberg equilibrium if the equilibrium exists. 
\end{prob}

\begin{rem}
Each stage of the OTZ-POSG is a two-player zero-sum Bayesian game with one-sided information \cite{zamir2020bayesian}. In the one-stage game, there are several normal-form games indexed by state $s\in S$. The leader has a probability distribution on the normal-form game while the follower knows the exact normal-form game they are playing. As the one-stage game builds the foundation for the multi-stage OTZ-POSG, we first discuss the one-stage game in the following section.
\end{rem}

\section{One-stage OTZ-POSGs}\label{sec:one_stage_game}
For the convenience of notation, we use $b$ to represent the belief, matrix $\Upsilon_{s^i} \in \mathbb{R}^{|A^L| \times |A^F|}$ to represent the reward matrix for state $s^i \in S$ and $\Upsilon^{[a^L,a^F]}_{s^i}$ to represent the element at row $a^L \in A^L$ and column $a^F \in A^F$.

\begin{defn} 
At the Stackelberg equilibrium, the value function of the leader is 
\begin{equation}\label{eq:value_func_leader}
    v^{L}(b) = \min_{\eta} \Big[ \sum_{s^i} b_{s^i} \big( \max_{\delta^{s^i}} \big[ \eta^T \Upsilon_{s^i} \delta^{s^i} \big] \big) \Big], 
\end{equation}
and the value function of the follower is $v^{F}(b,s^i) = \max_{\delta^{s^i}} \big[\hat{\eta}^T \Upsilon_{s^i} \delta^{s^i} \big], \forall s^i \in S$, where $\eta \in \mathbb{R}^{|A^L| \times 1}$ and $\delta^{s^i} \in \mathbb{R}^{|A^F| \times 1}$ are policies of players, and $\hat{\eta}$ is the solution of Equation (\ref{eq:value_func_leader}). 
\end{defn}

\begin{rem}
The value function $v^{L}(b)$ is the weighted average of the value function $v^{F}(b,s^i)$ over state $s^i$ and the weight is the belief $b$. Both value functions are well-defined because the Stackelberg equilibrium always exists. To see this, we fix the policy of the leader $\eta$ first and solve the value function $v^{F}(b,s^i)$ for any belief $b$ and state $s^i$. Because $\eta$ is in a bounded space, we can solve the value function $v^{L}(b)$ by taking the minimum value over a bounded space. The main challenge here is how to solve the optimal policy $\eta$ for the leader and represent both value functions concisely. To solve this issue, we first make an assumption on the reward function. 
\end{rem}

\begin{lem}\label{theo:equvalence}
Let $\Upsilon_{s^i} \in \mathbb{R}^{p \times q}$ and $\Theta_{s^i} \in \mathbb{R}^{p \times q}$ be reward matrices for state $s^i$, and they are related to each other by the relation $\Theta_{s^i} = \Upsilon_{s^i} + {c} \mathbf{1}_p \mathbf{1}_q^T, \forall s^i \in S$ where $c \in \mathbb{R}$ is a constant. Then, every mixed policy achieving the Stackelberg equilibrium for the matrix set $\{\Theta_{s^1},...,\Theta_{s^n}\}$ also constitutes a mixed policy at the Stackelberg equilibrium for the matrix set $\{\Upsilon_{s^1},...,\Upsilon_{s^n}\}$, and vice versa.
\end{lem}

The proof of the lemma is straightforward and thus omitted here. Through Lemma \ref{theo:equvalence}, we can assume that the reward function is lower bounded by a positive real value.

\begin{assum}\label{assum:possitivity}
The reward function of the game $\mathcal{G}$ is lower bounded by a positive real value, i.e., $\exists \ \underline{r}>0$ such that $\Upsilon_s^{[a^L,a^F]} \geq \underline{r}$ for all $s \in S, a^L \in A^L$, and $a^F \in A^F$. 
\end{assum}

\begin{thm}\label{theo:one_sided_game_linear_programming} 
For any one-stage OTZ-POSG, the policy $\eta$ that achieves the Stackelberg equilibrium defined by Equation (\ref{eq:value_func_leader}) can be solved by linear programming. 
\end{thm}

\begin{proof} Inspired by the work for normal-form games \cite{bacsar1998dynamic}, we define a function $f^i(\eta)$ for each state $s^i$ as $f^i(\eta) =  \max_{\delta^{s^i}} [\eta^T \Upsilon_{s^i} \delta^{s^i}]$. As the function $f^i(\eta)$ is the maximum value, we have $f^i(\eta) \geq \eta^T \Upsilon_{s^i} \delta^{s^i}$ for all $\delta^{s^i}$. It is equivalent to the inequality $f^i(\eta) \mathbf{1}_{|A^F|} \geq \Upsilon^T_{s^i} \eta$. Let's define a new variable $x^i = {\eta}/{f^i(\eta)}$. It is easy to see that $x^i$ are linearly dependent for all $i$. Define the scale factor between the vector  $x^1$ and $x^i$ to be $\rho^i =  x^1/{x^i}$. The policy $\eta$ and the corresponding value function $v^{L}(b)$ can be solved by the following optimization problem. 
\begin{equation}\label{eq:optimization_mixted_policy}
\begin{aligned}
    \min_{x^1,\rho^2,...,\rho^n} & \frac{b^T \rho}{({x^1})^T \mathbf{1}_{|A^L|}} \\
    \text{s.t.} & \ \Upsilon^T_{s^i} x^1 \leq \rho^i \mathbf{1}_{|A^F|}, \forall i \in \{1,...,n\}, \\
    & \ x^1 \geq \mathbf{0}_{|A^L|}, \ \ \rho^j > 0, \forall j \in \{2,...,n\}, \\
    & \ (x^1)^T \mathbf{1}_{|A^L|} \geq {1}/{\bar{f}}, 
\end{aligned}
\end{equation}
where $\rho = [\rho^1, \rho^2,...,\rho^n]^T \in \mathbb{R}^{n \times 1}$, $\rho^1 = 1$, $n = |S|$ and $\bar{f}$ is an upper bound of $f^i(\eta)$ for all $1 \leq i \leq n$.

The optimization problem given by Equation (\ref{eq:optimization_mixted_policy}) is a linear-fractional program as it is equivalent to 
\begin{equation}\label{eq:linear_fractional_programming}
\begin{aligned}
    \min_{z} & \ \frac{c^T z + \alpha}{{d}^T z} \\
    \text{s.t.} & \ \ \Gamma z \leq \beta, \\
    & \ \ d^T z \geq {1}/{\bar{f}}, \ \ z \geq \mathbf{0}_{|A^L|+n-1}, 
\end{aligned}
\end{equation}
where $\alpha = b_{s^1}$, 
\begin{equation*}
\begin{aligned}
    z = & \left[
    \begin{array}{l}
        x^1     \\
        \rho^2  \\
        \cdots  \\
        \rho^n
    \end{array}
    \right], 
    d = \left[ 
    \begin{array}{l}
        \mathbf{1}_{|A^L|} \\
        0       \\
        \cdots  \\
        0
    \end{array}
    \right], 
    c = \left[
    \begin{array}{l}
        \mathbf{0}_{|A^L|} \\
        b_{s^2} \\
        \cdots  \\
        b_{s^n}
    \end{array}
    \right], \\
    \Gamma = & \left[ 
    \begin{array}{rrrr}
        \Upsilon^T_{s^1},    & \mathbf{0}_{|A^F|},   & ...  \!     & \mathbf{0}_{|A^F|}  \\
        \Upsilon^T_{s^2},    & -\mathbf{1}_{|A^F|},  & ...  \!     & \mathbf{0}_{|A^F|}  \\ 
        \cdots               & \cdots                & ...  \!     & \cdots        \\
        \Upsilon^T_{s^n},    & \mathbf{0}_{|A^F|},   & ...  \!     & -\mathbf{1}_{|A^F|}
    \end{array}
    \right], 
    \beta = \left[ 
    \begin{array}{cc}
        \mathbf{1}_{|A^F|}  \\
        \mathbf{0}_{|A^F|}  \\
        \cdots              \\
        \mathbf{0}_{|A^F|}
    \end{array}
    \right].
\end{aligned}
\end{equation*}
The feasible region of the linear-fractional program is nonempty as the optimal policy of the leader exists and the value $d^T z = 1/f^1({\eta})$ is non-zero for all $\eta$. Hence, the linear-fractional program has a feasible solution. The feasible region of the linear-fractional program is not bounded. But we can restrict variables into a bounded space without changing the optimal solution of the original problem. To show this, we begin with the definition of the function $f^i(\eta)$. With Assumption \ref{assum:possitivity}, we have $f^i(\eta) \geq \underline{r}$. Hence, the variable $x^i \leq (1/\underline{r}) \mathbf{1}_{|A^L|}$ for all $i$. For the variable $\rho^i = f^i(\eta) / f^1(\eta)$, it is also upper bounded because $f^i(\eta)$ is upper bounded and $f^1(\eta)$ is lower bounded by $\underline{r}$. Meanwhile, we have $\rho^i > 0$ for all $i$. Hence, we can restrict the variables into a bounded space 
$\mathcal{U}=\{z | \mathbf{0}_{|A^L|} \leq x^1 \leq ({1}/{\underline{r}}) \mathbf{1}_{|A^L|}, 0 < \rho^i \leq {\bar{f}}/{\underline{r}},\forall i \}$ 
without changing the optimal solution.

Because the denominator ${d}^T z$ is lower bounded by a positive value, we can convert the linear-fractional program to a linear program \cite{charnes1962programming}. 
\begin{equation}\label{eq:optimization_linear_programming}
\begin{aligned}
    \min_{\mu, \lambda} & \ c^T \mu + \alpha \lambda\\
    \text{s.t.} & \ \Gamma \mu \leq \beta \lambda, \\
    & \ d^T \mu = 1,  \ \ \mathbf{0}_{|A^L|+n-1} \leq \mu \leq \bar{\mu}, \\
    & \ 0 < \lambda \leq \bar{f}, 
\end{aligned}
\end{equation}
where $\mu = \frac{z}{d^T z}$, $\lambda = \frac{1}{d^T z}$ and $\bar{\mu} = [\mathbf{1}^T_{|A^L|}, \bar{f}\mathbf{1}^T_{n-1}]^T$.
\end{proof}

\begin{rem}
We tighten the constraint when converting the linear-fractional program to a linear program because the extra constraint $\mathcal{U}$ is too loose. The extra constraint $\mathcal{U}$ is only used to guarantee that the feasible region is bounded. From the original optimization problem given by Equation (\ref{eq:optimization_mixted_policy}), we give a tighter constraint on the variable $\mu$ and $\lambda$ without changing solutions. By solving the linear program, the policy is $\eta = [\mu_{1},...,\mu_{|A^L|}]^T$ and the value $v^{L}(b)$ is $c^T \mu + \alpha \lambda$ for any $b$. Because the coefficient $c$ and $\alpha$ is uniquely determined by the belief $b$, the value function $v^{L}(b)$ is piece-wise linear for belief $b$.
\end{rem}

\begin{defn}
Let $\mathcal{D}$ represent the convex polyhedral set defined by the linear constraint of the linear program in Equation (\ref{eq:optimization_linear_programming}). A point $[\mu^T, \lambda]^T \in \mathbb{R}^{(|A^L| + n) \times 1}$ of the polyhedron $\mathcal{D}$ is called an extreme point if there exists a coefficient $[c^T, \alpha]^T \in \mathbb{R}^{(|A^L| + n) \times 1}$ such that $c^T \mu + \alpha \lambda < c^T \mu' + \alpha \lambda'$ for all $[(\mu')^T, \lambda']^T \neq  [\mu^T, \lambda]^T \in \mathcal{D}$. 
\end{defn}

For a linear objective function $c^T \mu + \alpha \lambda$ defined over a polyhedral convex set $\mathcal{D}$, the minimum value is taken only at extreme points of $\mathcal{D}$. Hence, a direct result from Theorem \ref{theo:one_sided_game_linear_programming} is the value function representation.

\begin{cor}
Let $\mathcal{V} = \{[(\mu^i)^T, \lambda^i]^T\}$ denote the set of all extreme points of the linear constraint in Equation (\ref{eq:optimization_linear_programming}). The value function of the leader can be represented concisely as ${v}^{L}({b}) = \min_i [b^T \theta(\mu^i, \lambda^i)]$ where $\theta(\mu, \lambda) = [\lambda, \mu_{|A^L|+1},...,\mu_{|A^L|+n}]^T$ is a vector extracting elements from $\mu$ and $\lambda$ corresponding to the nonzero entries of the coefficient $[c^T,\alpha]^T$. 
\end{cor}

\begin{rem}
All extreme points of the polyhedral convex set $\mathcal{D}$ can be solved using the algorithm proposed in \cite{balinski1961algorithm}. The value function of the leader is piece-wise linear and convex for the belief $b$. Correspondingly, the value function ${v}^{F}({b},s^i)$ is piece-wise constant for belief $b$. 
\end{rem}

\begin{exmp}\label{exam:one_horizon_game}
We consider a one-stage OTZ-POSG where the state space is $S=\{s^1, s^2\}$, the action set of the leader is $A^L = \{a^L_1, a^L_2\}$, the action set of the follower is $A^F = \{a^F_1, a^F_2\}$, and the reward matrix is $\Upsilon_{s^1} = [4, 2; 2, 7]$ for state $s^1$ and $\Upsilon_{s^2} = [8, 6; 3, 4]$ for state $s^2$. Extreme points and the corresponding policies $\eta$ derived from this game are listed in the following table. 
\begin{equation*}
\begin{array}{|c|cc|cc||c|cc|cc|c|cc|c|}
    \text{} & \hat{\eta}_1 & \hat{\eta}_2 & \lambda & \mu_3 \\
    \hline
    1 & 0.000  & 1.000  & 7.000  & 4.000 \\
    2 & 0.333  & 0.667  & 5.333  & 4.667 \\
    3 & 0.714  & 0.286 &  3.429  & 6.571 
\end{array}
\end{equation*}

The piece-wise linear value function ${v}^{L}({b})$ and piece-wise constant value function $v^F(b,s^i)$ are shown in Fig.~\ref{fig:policy_player2}. Three extreme points are found and vectors $\theta(\mu^i, \lambda^i)$ are shown by dotted lines. For different belief $b$, an extreme point is chosen by the $\min$ operator. 

\begin{figure}
\centering
\includegraphics[width=0.95\linewidth]{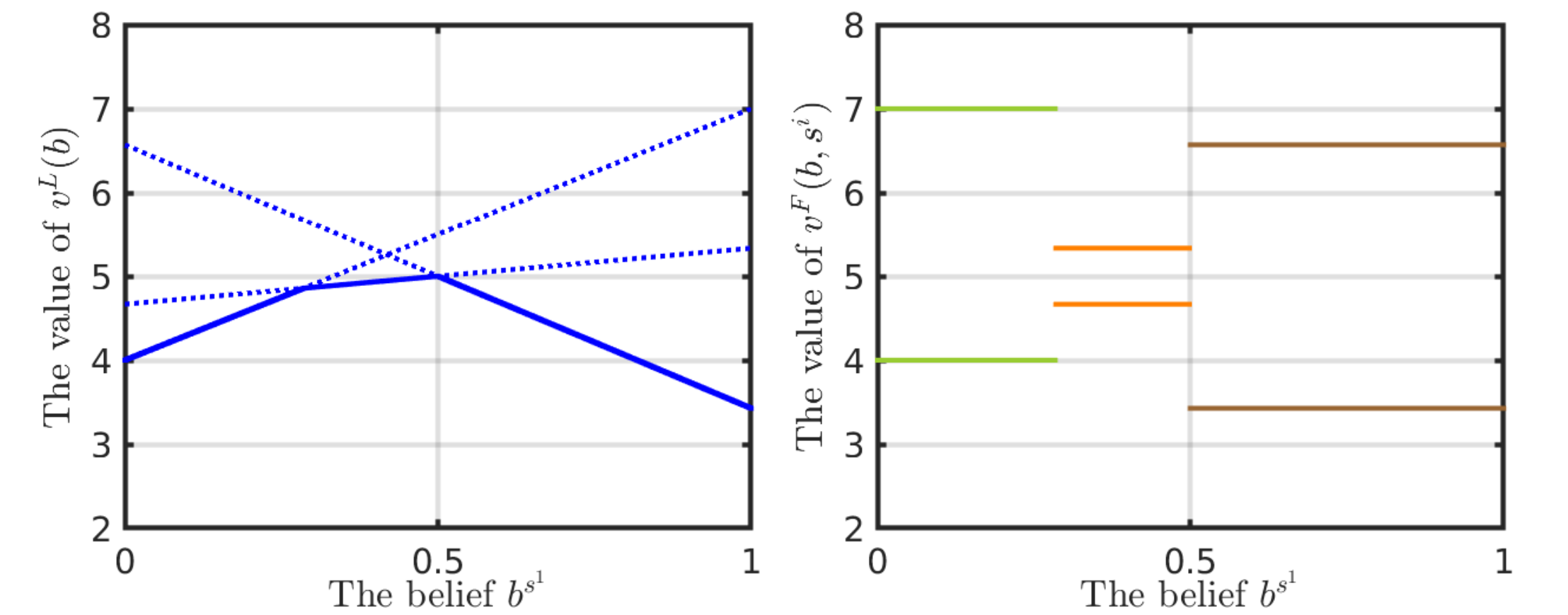}
\caption{The piece-wise linear value function $v^L(b)$ and the piece-wise constant value function $v^F(b,s^i)$.}
\label{fig:policy_player2}
\end{figure}
\end{exmp}

\section{Multi-stage OTZ-POSGs}\label{sec:multiple_stage}
For the OTZ-POSG with multiple stages, the equilibrium can be solved by dynamic programming and the total rewards can be solved through value iteration.  Technically, at stage $t$, the value function of the leader achieving the Stackelberg equilibrium is
\begin{equation}\label{eq:mutl_vf_leader}
    {v}^{L}({b}_{t}) = \min_{\eta_t} \Big[ \sum_{s_t} b_{t}^{s_t} \big( \max_{\delta^{s_t}_t} \big[ \eta_t^T (\Upsilon_{s_t} + \Phi_{s_t}) \delta^{s_t}_t \big] \big) \Big], 
\end{equation}
where $\eta_t \in \mathbb{R}^{|A^L| \times 1}$ and $\delta^{s_t}_t \in \mathbb{R}^{|A^F| \times 1}$ are policies. The matrix $\Phi_{s_t} \in \mathbb{R}^{|A^L| \times |A^F|}$ represents the future reward for each joint action $a_t = [a_t^L, a_t^F]$. The element at row $a_t^L$ and column $a_t^F$ is $\Phi_{s_t}^{a_t} = \sum_{s_{t+1}} T^{s_{t+1}}_{s_t,a_t} \sum_{o_{t+1}} \Xi^{o_{t+1}}_{s_{t+1}} {v}^{F}(b_{t+1},s_{t+1})$. For any $s_t \in S$, the value function of the follower is ${v}^{F}(b_{t},s_t) = \hat{\eta}_t^T (\Upsilon_{s_t} + \Phi_{s_t}) \hat{\delta}^{s_t}_t$ 
where $\hat{\eta}_t$ and $\hat{\delta}^{s_t}_t$ are solutions of Equation (\ref{eq:mutl_vf_leader}).

The main challenge is the information leakage issue which also appeared in the repeated games with incomplete information \cite{aumann1995repeated}. As Fig.~\ref{fig:information_flow} shows, the action taken by the follower reveals the state information because the policy of the follower is state-dependent. If the follower takes full advantage of the private information, the follower reveals the state information to the leader. Then, the leader can infer more state information and achieves more rewards in the future.

To solve this issue, we consider the $\epsilon$-Stackelberg equilibrium. The basic idea is to sacrifice certain rewards in the current stage for more future rewards (from the follower's perspective). First, we define a sub-optimal policy for the follower.

\begin{defn} A policy $\tilde{\delta}_t^{s_t} \in \mathbb{R}^{|A^F| \times 1}$ is said to be a $\epsilon$-sacrifice policy of $\hat{\delta}_t^{s_t}$ if $\hat{\eta}_t^T (\Upsilon_{s_t}  + \Phi_{s_t})(\hat{\delta}^{s_t}_t - \tilde{\delta}^{s_t}_t) \leq \epsilon$ for all state $s_t \in S$. 
\end{defn}

By adopting the $\epsilon$-sacrifice policy, the follower guarantees that the sacrificed reward is bounded by $\epsilon$. Meanwhile, by keeping the $\epsilon$-sacrifice policy private, the follower prevents the leader from inferring the state information for future stages. Inspired by the work \cite{bacsar1998dynamic}, we fit the concept of $\epsilon$-Stackelberg equilibrium to our problem as follows.

\begin{defn} Given the total rewards for both players, a pair of policies $[\hat{\pi}^{L}, \hat{\pi}^{F}]$ forms a Stackelberg equilibrium if it satisfies the following conditions, 
\begin{equation}
\begin{array}{lll}
    v^L_{\hat{\pi}^{L}, \sigma(\hat{\pi}^{L})}(b_0)  & \leq v^L_{\pi^{L}, \sigma({\pi}^{L})}(b_0) + \epsilon, &\forall \pi^{L}, \\
    v^F_{\hat{\pi}^{L}, \hat{\pi}^{F}}(b_0,s_0) & \geq v^F_{\hat{\pi}^{L}, \pi^{F}}(b_0, s_0) - \epsilon, &\forall \pi^{F}, 
\end{array}
\end{equation}
where $\sigma({\pi}^{L})$ is a reaction function of the follower defined by $\sigma({\pi}^{L}) = \argmax_{\pi^F} v^F_{{\pi}^{L}, {\pi}^{F}}(b_0,s_0)$.
\end{defn}

To show the $\epsilon$-Stackelberg equilibrium, we first calculate the matrix $\Phi_{s_t}$ from the value function ${v}^{F}(b_{t+1},s_{t+1})$. Although the value function is piece-wise constant, its value relies on the belief $b_{t+1}$. We have to represent it as a function of belief $b_t$. To solve this issue, we propose a belief space partition approach.

\begin{defn}
Given a belief space $\Delta$, a partition of the belief space $\Delta$ is defined as $\Lambda = \{\Delta_1,...,\Delta_m\}$ where $\Delta = \cup_i \Delta_i$ and $\Delta_i \cap \Delta_j = \emptyset, \forall i \ne j$.
\end{defn}

In our proposed belief space partition approach, each partition $\Delta_i$ is represented as $\Delta_i =\{b \in \Delta|\Pi^i b \leq \mathbf{0}_{l_i}\}$ where $\Pi^i \in \mathbb{R}^{|l_i| \times |S|}$ can be constructed iteratively. To illustrate the space partition approach, we begin with the stage $t+1$ and assume that the belief space partition $\Lambda^{t+1} = \{\Delta^{t+1}_1,..., \Delta^{t+1}_{m_{t+1}}\}$ is given. A belief $b_{t+1}$ belongs to the set $\Delta^{t+1}_i$ if $\Pi^i_{t+1} b_{t+1} \leq \mathbf{0}_{l_i^{t+1}}$. For each joint action $a_t$ and observation $o_{t+1}$, plugging in the belief $b_{t+1}$ from Equation (\ref{eq:belief_updating}), we can convert the linear constraint $\Pi^i_{t+1} b_{t+1} \leq \mathbf{0}_{l_i^{t+1}}$ into $\bar{\Pi}^i_{t} b_{t} \leq \mathbf{0}_{l_i^{t+1}}$. The partition of the belief $b_t$ is $\bar{\Lambda}^{t}_{a_t, o_{t+1}} = \{\bar{\Delta}^{t}_i, ..., \bar{\Delta}^{t}_{m_{t+1}}\}$ with $\bar{\Delta}^{t}_i$ defined as $\bar{\Delta}^t_i = \{b_t|\bar{\Pi}^i_{t} b_{t} \leq \mathbf{0}_{l_i^{t+1}}\}$. Combing partitions of all joint actions $a_t$ and observations $o_{t+1}$, we have a finer partition of the belief space $\bar{\Lambda}^t = \{\cap_{a_t, o_{t+1}} \bar{\Delta}^t_{a_t,o_{t+1}} | \bar{\Delta}^t_{a_t,o_{t+1}} \in \bar{\Lambda}^{t}_{a_t, o_{t+1}} \}$. The constraint matrix for the intersection can be achieved by concatenating the constraint matrix $\bar{\Pi}^i_{t}$ together. In each region ${\Delta}^t_i$ of set $\bar{\Lambda}^t$, the value ${v}^{F}(b_{t+1},s_{t+1})$ is a constant value. Hence, matrices $\{\Phi_{s_t}, s_t \in S\}$ are constant and can be calculated. Based on the matrices $\{\Phi_{s_t}, s_t \in S\}$, a new linear program can be formulated and all extreme points can be founded. The region ${\Delta}^t_i$ can be further partitioned using these extreme points. In each partition, the value function ${v}^{L}(b_{t})$ is linear and the value function ${v}^{F}(b_{t},s_{t})$ is constant.

Through partitioning the belief space, we can calculate the value function iteratively. When the follower adopts a $\epsilon$-sacrifice policy, the $\epsilon$-Stackelberg equilibrium is achieved.

\begin{assum}\label{assum:sacrifice_policy}
The fact that the follower adopts a $\epsilon$-sacrifice policy is common knowledge for both players, but the value $\epsilon$ and the $\epsilon$-sacrifice policy are private. 
\end{assum}

\begin{rem}
Through Assumption \ref{assum:sacrifice_policy}, we assert that the leader only updates the belief state using observations. It is because the value $\epsilon$ and the $\epsilon$-sacrifice policies are private information. It is nontrivial for the leader to infer the $\epsilon$-sacrifice policy in finite stages.
\end{rem}

\begin{thm}\label{thm:epsilon}
Given a finite horizon $h$ and a positive real value $\epsilon$, the OTZ-POSG achieves an $\epsilon$-Stackelberg equilibrium if the follower adopts a $\frac{\epsilon}{h+1}$-sacrifice policy. 
\end{thm}

\begin{proof} 
We prove this theorem by induction. Let $\tilde{v}^{L}(b_{t})$ and $\tilde{v}^{F}(b_{t},s_t)$ denote value functions when the follower adopts a ${\epsilon}/{(h+1)}$-sacrifice policy. At stage $h$, it is straightforward to verify the inequality $|{v}^{F}({b}_{h},s_h)-\tilde{v}^{F}({b}_{h},s_h)|\leq {\epsilon}/{(h+1)}$ for any state $s_h$. At stage $t+1$, we assume that $|{v}^{F}(b_{t+1},s_{t+1}) - \tilde{v}^{F}(b_{t+1},s_{t+1})|\leq (h-t){\epsilon}/{(h+1)}$ for any state $s_{t+1}$. Then, at stage $t$, we have $|\Phi_{s_t}^{a_t} - \bar{\Phi}_{s_t}^{a_t}|\leq {(h-t)\epsilon}/{(h+1)}$ where $\bar{\Phi}_{s_t}$ is the matrix calculated from value function $\tilde{v}^{F}(b_{t+1},s_{t+1})$. It is easy to check that $|{v}^{F}(b_{t},s_{t}) - \bar{v}^{F}(b_{t},s_{t})|\leq {(h-t)\epsilon}/{(h+1)}$ where $\bar{v}^{F}(b_{t},s_{t})$ is the value function derived with matrices $\{\bar{\Phi}_{s_t}, s_t \in S\}$. After adopting a $ {\epsilon}/({h+1})$-sacrifice policy at stage $t$, the value sacrificed by the follower is bounded, i.e., $|\bar{v}^{F}(b_{t},s_{t}) - \tilde{v}^{F}(b_{t},s_{t})| \leq {\epsilon}/({h+1})$. Hence, the total distance $|{v}^{F}(b_{t},s_{t}) - \tilde{v}^{F}(b_{t},s_{t})|$ is bounded by $|{v}^{F}(b_{t},s_{t}) - \bar{v}^{F}(b_{t},s_{t})| + |\bar{v}^{F}(b_{t},s_{t}) - \tilde{v}^{F}(b_{t},s_{t})| = (h-t+1)\epsilon/({h+1})$ for any state $s_t \in S$. By induction, we have $|{v}^{F}({b}_{0},s_{0})-\tilde{v}^{F}({b}_{0},s_{0})|\leq \epsilon$. Because ${v}^{L}({b}_{0}) = \sum_{s_0} b_0^{s_0} {v}^{F}({b}_{0},s_0)$ and $\tilde{v}^{L}({b}_{0}) = \sum_{s_0} b_0^{s_0} \tilde{v}^{F}({b}_{0},s_0)$, we have $|{v}^{L}({b}_{0})-\tilde{v}^{L}({b}_{0})|\leq \epsilon$. As a consequence, the $\epsilon$-Stackelberg equilibrium is achieved because we have $\tilde{v}^{F}({b}_{0},s_{0}) \geq {v}^{F}({b}_{0},s_{0})-\epsilon$ and $\tilde{v}^{L}({b}_{0}) \leq {v}^{L}({b}_{0}) \leq {v}^{L}({b}_{0}) + \epsilon$. 
\end{proof}

\begin{thm}\label{theo:piece_wise_lienar_constant}
The value function of the leader is piece-wise linear and the value function of the follower is piece-wise constant for any stage $t$. 
\end{thm}

The theorem is a direct result of the belief space partition approach and the proof is omitted here.

\begin{rem}
By Theorem \ref{thm:epsilon} and through the belief space partition approach, the $\epsilon$-Stackelberg equilibrium is achieved and policies at the equilibrium are solved iteratively. Hence, Problem \ref{problem:main_problem} is solved.
\end{rem}

\begin{rem} 
From Theorem \ref{theo:piece_wise_lienar_constant}, the belief space is partitioned finer and finer. In the worst case, there are 
$$\varkappa =\frac{(|A^L|+|S|+|A^F||S|)!}{(|A^L|+|S|-1)!(|A^F||S|+1)!}$$
extreme points for each linear program. The total number of partitions grows double exponentially with respect to the planning horizon $h$, i.e. $\mathcal{O}(\varkappa^{(|A^L||A^F||O|)^h} )$, which is a potential bottleneck of the value function calculation. However, an approximation algorithm with performance guarantees is nontrivial to develop because the value function of the follower is piece-wise constant. To approximate this value function, evaluating the boundary is inevitable. 
\end{rem}

\begin{exmp}\label{exam:example3}
Consider a OTZ-POSG model where the state space, the action space and the reward function are defined in Example \ref{exam:one_horizon_game}. The observation set is $O = \{o^1,o^2\}$ and the transition function is  
\begin{equation*}
\begin{array}{||c|cc||c|cc||c|cc||c|cc|}
    a^{1,1} & s_1 & s_2  & a^{1,2} & s_1 & s_2\\
    \hline
    s_1     & 0.3 & 0.7  & s_1     & 0.0 & 1.0 \\
    s_2     & 0.9 & 0.1  & s_2     & 0.8 & 0.2 \\
    \hline
    a^{2,1} & s_1 & s_2 & a^{2,2}  & s_1 & s_2 \\
    \hline
    s_1     & 0.8 & 0.2 & s_1      & 0.5 & 0.5 \\
    s_2     & 0.1 & 0.9 &  s_2     & 0.0 & 1.0 
\end{array}, 
\end{equation*}
where $a^{i,j}$ represents the joint action $[a_i^L,a_j^F]$. The observation probability is $\Xi_{s^1} = [0.6,0.4]$ and $\Xi_{s^2} = [0.1,0.9]$. We assume that the value functions at stage $t+1$ is given by Fig.~\ref{fig:policy_player2}. The belief space is partitioned into three regions represented by ${\Delta}^{t+1}_i = \{b_{t+1}|{\Pi}^i_{t+1} b_{t+1} \leq \mathbf{0}_2\}$ where ${\Pi}^1_{t+1}=[1.67, -0.67;$ $3.57,-2.57]$, ${\Pi}^2_{t+1} = [-1.67, 0.67;1.91,-1.91]$ and ${\Pi}^3_{t+1} = [-3.57, 2.57; -1.91, 1.91]$. For each joint action and observation, we can convert the linear constraint into the form of $\bar{\Pi}^i_{t} b_{t} \leq \mathbf{0}_2$. In this process, the number of partitions may be reduced. For the joint action $a^{1,1}_{t}$ and the observation $o^1_{t+1}$, the constraint matrices $\bar{\Pi}^i_{t}$ are $\bar{\Pi}^1_{t} = [0.25,0.89;0.46,1.90]$, $\bar{\Pi}^2_{t} = [-0.25, -0.89;0.21,1.01]$ and $\bar{\Pi}^3_{t} = [-0.46,-1.90;$ $-0.21,-1.01]$. Among them, only one partition $\bar{\Delta}^t_3 = \{b_t|\bar{\Pi}^3_{t} b_{t} \leq \mathbf{0}_2\}$ is active. For joint action $a^{1,1}_{t}$ and observation $o^2_{t+1}$, the constraint matrices are $\bar{\Pi}^1_{t} = [-0.22,0.54;-1.19,1.05]$, $\bar{\Pi}^2_{t} = [0.22,-0.54; -0.97,0.51]$ and $\bar{\Pi}^3_{t} = [1.19, -1.05;$ $0.97,-0.51]$. All these partitions are active and the corresponding regions are shown in the left figure of Fig.~\ref{fig:beleifspace_partition}. The value $\sum_{s_{t+1}} \Xi^{o_{t+1}}_{s_{t+1}} T^{s_{t+1}}_{s_t,a_t} {v}^{F}(b_{t+1},s_{t+1})$ when $o_{t+1} = o^2_{t+1}$ and $a_t=a_t^{1,1}$ is shown in the right figure of Fig.~\ref{fig:beleifspace_partition}. The red dash-dot lines represent the boundaries of partitions. After combing partitions of all joint actions and observations, we have a finer partition of the belief space. In each partition, we can calculate all candidate $\alpha$-vectors because the matrix $\Phi_{s_t}$ is constant. The value function at stage $t$ is shown in Fig.~\ref{fig:dynamic_partitions}. The dotted lines represent candidate $\alpha$-vectors in each partition and solid lines represent the picked $\alpha$-vectors for the value function.

\begin{figure}
\includegraphics[width=0.95\linewidth]{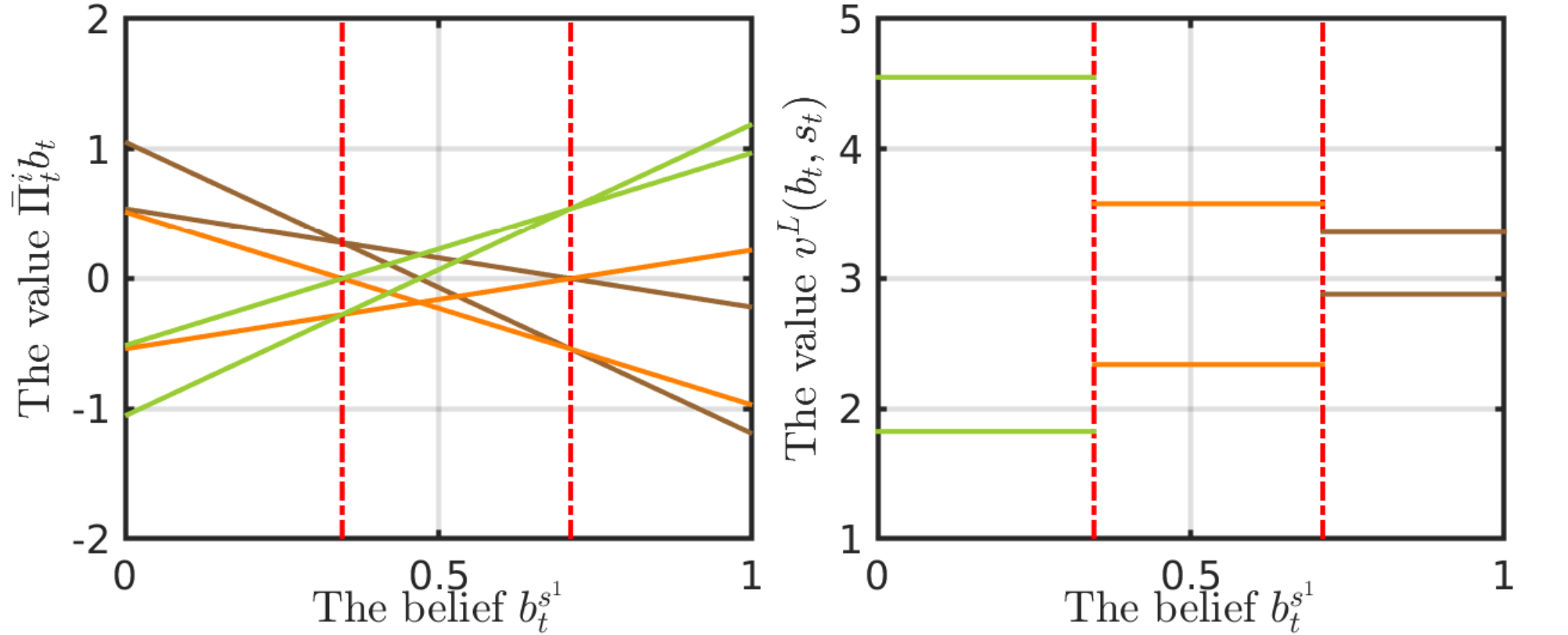}
\caption{The partition of the belief space at stage $t+1$ for the joint action $a_t^{1,1}$ and observation $o^2_{t+1}$.}
\label{fig:beleifspace_partition}
\end{figure}

\begin{figure}
\includegraphics[width=0.95\linewidth]{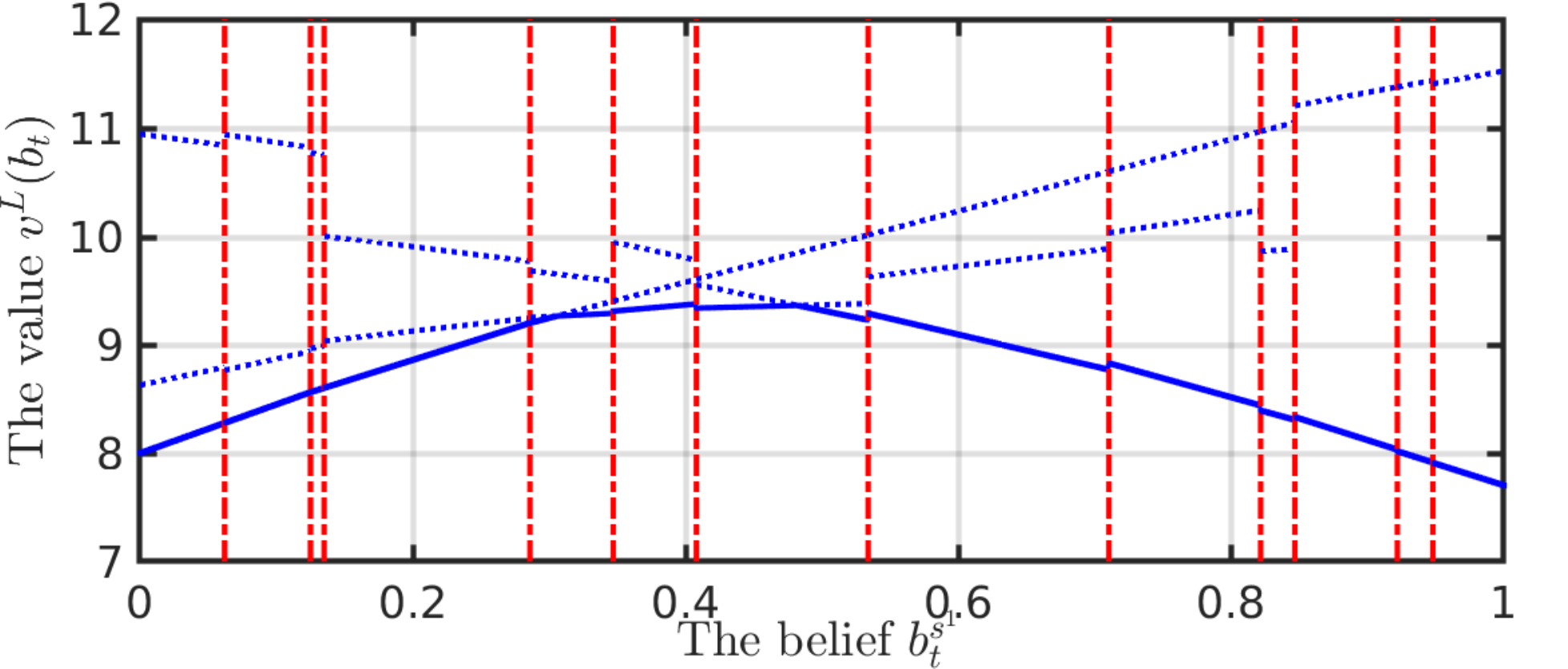}
\caption{The partition of the belief space at stage $t$ and the corresponding value function of the leader.}
\label{fig:dynamic_partitions}
\end{figure}
\end{exmp}

\begin{rem}
From Fig.~\ref{fig:dynamic_partitions}, we see that the value function of the leader is piece-wise linear but not continuous. It is why existing planning algorithms for the POMDP model do not work for the OTZ-POSG. The POMDP can be treated as a special case of the OTZ-POSG. Hence, the proposed space partition algorithm can be applied on the POMDP model. 
\end{rem}

\section{Conclusion}\label{sec:conc}
In this paper, we considered the policy design problem for turn-based OTZ-POSGs with public actions. We proved the existence of the Stackelberg equilibrium for the one-stage OTZ-POSG and shown that, in each stage, the game can be converted into a linear-fractional programming problem, and therefore, solved by linear programming. By enumerating all extreme points of the linear program, we have shown that the value function of the leader is piece-wise linear and the value function of the follower is piece-wise constant. For the finite-horizon POSG, we have proved that the $\epsilon$-Stackelberg equilibrium is achieved. This study will pave the way towards a formal and systematic design theory for problems such as cyber-security defense when actions are public. One bottleneck of the proposed approach is the high computation complexity as the number of partitions grows very fast for the planning horizon. To reduce the complexity will be one of our further work.

\bibliographystyle{plain}
\bibliography{OTZPOSG}

\begin{thebibliography}{10}

\bibitem{aberdeen2007policy}
Douglas Aberdeen, Olivier Buffet, and Owen Thomas.
\newblock Policy-gradients for {PSR}s and {POMDP}s.
\newblock In {\em Artificial Intelligence and Statistics}, pages 3--10, 2007.

\bibitem{ahmadi2018partially}
Mohamadreza Ahmadi, Murat Cubuktepe, Nils Jansen, Sebastian Junges,
  Joost-Pieter Katoen, and Ufuk Topcu.
\newblock The partially observable games we play for cyber deception.
\newblock {\em arXiv preprint arXiv:1810.00092}, 2018.

\bibitem{aumann1995repeated}
Robert~J Aumann, Michael Maschler, and Richard~E Stearns.
\newblock {\em Repeated games with incomplete information}.
\newblock MIT press, 1995.

\bibitem{balinski1961algorithm}
Michel~L Balinski.
\newblock An algorithm for finding all vertices of convex polyhedral sets.
\newblock {\em Journal of the Society for Industrial and Applied Mathematics},
  9(1):72--88, 1961.

\bibitem{bacsar1998dynamic}
Tamer Ba{\c{s}}ar and Geert~Jan Olsder.
\newblock {\em Dynamic noncooperative game theory}.
\newblock SIAM, 1998.

\bibitem{breton1988sequential}
Michele Breton, Abderrahmane Alj, and Alain Haurie.
\newblock Sequential {S}tackelberg equilibria in two-person games.
\newblock {\em JOTA}, 59(1):71--97, 1988.

\bibitem{buczak2015survey}
Anna~L Buczak and Erhan Guven.
\newblock A survey of data mining and machine learning methods for cyber
  security intrusion detection.
\newblock {\em IEEE Communications surveys $\&$ tutorials}, 18(2):1153--1176,
  2015.

\bibitem{charnes1962programming}
Abraham Charnes and William~W Cooper.
\newblock Programming with linear fractional functionals.
\newblock {\em Naval Research logistics quarterly}, 9(3-4):181--186, 1962.

\bibitem{conitzer2006computing}
Vincent Conitzer and Tuomas Sandholm.
\newblock Computing the optimal strategy to commit to.
\newblock In {\em Proceedings of the 7th ACM conference on Electronic
  commerce}, pages 82--90, 2006.

\bibitem{do2017game}
Cuong~T Do, Nguyen~H Tran, Choongseon Hong, Charles~A Kamhoua, Kevin~A Kwiat,
  Erik Blasch, and et~al.
\newblock Game theory for cyber security and privacy.
\newblock {\em ACM Computing Surveys (CSUR)}, 50(2):1--37, 2017.

\bibitem{emery2004approximate}
Rosemary Emery-Montemerlo, Geoff Gordon, Jeff Schneider, and Sebastian Thrun.
\newblock Approximate solutions for partially observable stochastic games with
  common payoffs.
\newblock In {\em Proceedings of the Third International Joint Conference on
  AAMAS, 2004. AAMAS 2004.}, pages 136--143. IEEE, 2004.

\bibitem{fallah2008puzzle}
Mehran Fallah.
\newblock A puzzle-based defense strategy against flooding attacks using game
  theory.
\newblock {\em IEEE transactions on dependable and secure computing},
  7(1):5--19, 2008.

\bibitem{hansen2004dynamic}
Eric~A Hansen, Daniel~S Bernstein, and Shlomo Zilberstein.
\newblock Dynamic programming for partially observable stochastic games.
\newblock In {\em AAAI}, volume~4, pages 709--715, 2004.

\bibitem{horak2016point}
Karel Hor{\'a}k and Branislav Bo{\v{s}}ansk{\`y}.
\newblock A point-based approximate algorithm for one-sided partially
  observable pursuit-evasion games.
\newblock In {\em International Conference on Decision and Game Theory for
  Security}, pages 435--454. Springer, 2016.

\bibitem{horak2019solving}
Karel Hor{\'a}k and Branislav Bo{\v{s}}ansk{\`y}.
\newblock Solving partially observable stochastic games with public
  observations.
\newblock In {\em Proceedings of the AAAI Conference on Artificial
  Intelligence}, volume~33, pages 2029--2036, 2019.

\bibitem{horak2017heuristic}
Karel Hor{\'a}k, Branislav Bo{\v{s}}ansk{\`y}, and Michal
  P{\v{e}}chou{\v{c}}ek.
\newblock Heuristic search value iteration for one-sided partially observable
  stochastic games.
\newblock In {\em Thirty-First AAAI Conference on Artificial Intelligence},
  2017.

\bibitem{horak2019optimizing}
Karel Hor{\'a}k, Branislav Bo{\v{s}}ansk{\`y}, Petr Tom{\'a}{\v{s}}ek,
  Christopher Kiekintveld, and Charles Kamhoua.
\newblock Optimizing honeypot strategies against dynamic lateral movement using
  partially observable stochastic games.
\newblock {\em Computers $\&$ Security}, 87:101579, 2019.

\bibitem{huang2017markov}
Jianming Huang, Hengwei Zhang, and Jindong Wang.
\newblock Markov evolutionary games for network defense strategy selection.
\newblock {\em IEEE Access}, 5:19505--19516, 2017.

\bibitem{humayed2017cyber}
Abdulmalik Humayed, Jingqiang Lin, Fengjun Li, and Bo~Luo.
\newblock Cyber-physical systems security - a survey.
\newblock {\em IEEE Internet of Things Journal}, 4(6):1802--1831, 2017.

\bibitem{kumar2009dynamic}
Akshat Kumar and Shlomo Zilberstein.
\newblock Dynamic programming approximations for partially observable
  stochastic games.
\newblock In {\em Proceedings of the Twenty-Second International FLAIRS
  Conference}, page 547–552, 2009.

\bibitem{lalropuia2019modeling}
KC~Lalropuia and Vandana Gupta.
\newblock Modeling cyber-physical attacks based on stochastic game and {M}arkov
  processes.
\newblock {\em Reliability Engineering $\&$ System Safety}, 181:28--37, 2019.

\bibitem{luo2009game}
Yi~Luo, Ferenc Szidarovszky, Youssif Al-Nashif, and Salim Hariri.
\newblock Game tree based partially observable stochastic game model for
  intrusion defense systems ({IDS}).
\newblock In {\em IIE Annual Conference. Proceedings}, page 880. IISE, 2009.

\bibitem{salloum2020machine}
Said~A Salloum, Muhammad Alshurideh, Ashraf Elnagar, and Khaled Shaalan.
\newblock Machine learning and deep learning techniques for cybersecurity: a
  review.
\newblock In {\em Joint European-US Workshop on Applications of Invariance in
  Computer Vision}, pages 50--57. Springer, 2020.

\bibitem{tomavsek2020using}
Petr Tom{\'a}{\v{s}}ek, Branislav Bo{\v{s}}ansk{\`y}, and Thanh~H Nguyen.
\newblock Using one-sided partially observable stochastic games for solving
  zero-sum security games with sequential attacks.
\newblock In {\em International Conference on Decision and Game Theory for
  Security}, pages 385--404. Springer, 2020.

\bibitem{tsemogne2020partially}
Olivier Tsemogne, Yezekael Hayel, Charles Kamhoua, and Gabriel Deugoue.
\newblock Partially observable stochastic games for cyber deception against
  network epidemic.
\newblock In {\em International Conference on Decision and Game Theory for
  Security}, pages 312--325. Springer, 2020.

\bibitem{wang2019players}
Xinrun Wang, Milind Tambe, Branislav Bo{\v{s}}ansk{\`y}, and Bo~An.
\newblock When players affect target values: Modeling and solving dynamic
  partially observable security games.
\newblock In {\em International Conference on Decision and Game Theory for
  Security}, pages 542--562. Springer, 2019.

\bibitem{zamir2020bayesian}
Shmuel Zamir.
\newblock {\em Bayesian games: Games with incomplete information}.
\newblock Springer, 2020.

\end{thebibliography}
\end{document}